\documentclass[lettersize,journal]{IEEEtran}
\IEEEoverridecommandlockouts

\usepackage{amssymb}
\usepackage[cmex10]{amsmath}
\usepackage{stfloats}
\usepackage{graphicx}
\usepackage{subfigure}
\usepackage{tabularx}
\usepackage{epsfig,epsf,color,balance,cite}
\usepackage{verbatim}
\usepackage{url}
\usepackage{bm}
\usepackage{booktabs}
\usepackage{siunitx}

\usepackage{amsmath}
\allowdisplaybreaks[2]
\usepackage{ntheorem}
\newtheorem{theorem}{Theorem}

\newtheorem*{proof}{Proof}

\usepackage{algorithm}
\usepackage{algorithmic}

\usepackage{color}
\definecolor{myc1}{rgb}{0,0,0}

\begin{document}

\title{Energy-Efficient Probabilistic Semantic Communication Over Visible Light Networks\\ With Rate Splitting}

\author{Zhouxiang Zhao, 
        Zhaohui Yang, 
        Chen Zhu,
        Xin Tong,
        and Zhaoyang Zhang,~\IEEEmembership{Senior Member,~IEEE}
\thanks{Zhouxiang Zhao, Zhaohui Yang, and Zhaoyang Zhang are with the College of Information Science and Electronic Engineering, Zhejiang University, and also with Zhejiang Provincial Key Laboratory of Info. Proc., Commun. \& Netw. (IPCAN), Hangzhou 310027, China (e-mails: \{zhouxiangzhao, yang\_zhaohui, ning\_ming\}@zju.edu.cn).}
\thanks{Chen Zhu is with School of Communication Engineering, Hangzhou Dianzi University, Hangzhou 310018, China, and also with Polytechnic Institute, Zhejiang University, Hangzhou 310015, China (e-mail: zhuc@zju.edu.cn).}
\thanks{Xin~Tong is with the Department of Electrical Engineering, Chalmers University of Technology, Gothenburg, Sweden (e-mail: xinto@chalmers.se).}
}

\maketitle

\begin{abstract}
Visible light communication (VLC) is emerging as a key technology for future wireless communication systems due to its unique physical-layer advantages over traditional radio-frequency (RF)-based systems. However, its integration with higher-layer techniques, such as semantic communication, remains underexplored. This paper investigates the energy efficiency maximization problem in a resource-constrained VLC-based probabilistic semantic communication (PSCom) system. In the considered model, light-emitting diode (LED) transmitters perform semantic compression to reduce data size, which incurs additional computation overhead. The compressed semantic information is transmitted to the users for semantic inference using a shared knowledge base that requires periodic updates to ensure synchronization. In the PSCom system, the knowledge base is represented by probabilistic graphs. To enable simultaneous transmission of both knowledge and information data, rate splitting multiple access (RSMA) is employed. The optimization problem focuses on maximizing energy efficiency by jointly optimizing transmit beamforming, direct current (DC) bias, common rate allocation, and semantic compression ratio, while accounting for both communication and computation costs. To solve this problem, an alternating optimization algorithm based on successive convex approximation (SCA) and Dinkelbach method is developed. Simulation results demonstrate the effectiveness of the proposed approach.
\end{abstract}

\begin{IEEEkeywords}
Visible light communications, semantic communications, rate splitting, energy efficiency.
\end{IEEEkeywords}

\IEEEpeerreviewmaketitle

\section{Introduction}
\IEEEPARstart{D}{riven} by the increasing need for high-speed, energy-efficient, and reliable data transfer, the demand for advanced communication systems has grown exponentially in recent years \cite{7833457,8869705,10061665}. Traditional communication methods, while effective, face limitations in terms of energy efficiency, bandwidth, and handling complex data environments \cite{xu2023edge,10283723}. Emerging technologies in areas such as the Internet of Things (IoT), smart cities, and autonomous systems place new demands on communication networks, pushing current radio-frequency (RF)-based systems to their limits \cite{6740844}. Moreover, the growing importance of low-latency, high-throughput data transfer further emphasizes the need for innovation in communication technologies \cite{zhao2025agentic,8972897}. In light of these developments, it is imperative to design communication systems that can support higher data rates while simultaneously reducing energy consumption and maintaining reliable performance.

Recently, visible light communication (VLC) has emerged as a promising candidate to address these challenges \cite{7239528}. VLC leverages the visible light spectrum (380 – 780 nm) for data transmission, a spectrum that is unregulated and abundant compared to the increasingly crowded RF spectrum \cite{8698841}. This approach utilizes existing infrastructure, such as light-emitting diodes (LEDs) used for illumination, to simultaneously provide lighting and data communication, creating an energy-efficient dual-purpose system \cite{7072557,7839766}. Furthermore, VLC is immune to electromagnetic interference, making it ideal for environments such as hospitals or aircraft, where RF signals may cause disruptions. Another significant advantage is its inherent security, since visible light does not penetrate walls, providing a natural boundary against eavesdropping \cite{9070153}.

Complementing VLC’s physical-layer advantages is the emerging paradigm of semantic communication, which aims to optimize data transfer not just by transmitting bits, but by conveying meaning \cite{gunduz2022beyond,qin2021semantic,10915662}. This approach focuses on the semantic content of transmitted information, ensuring that only the most relevant data is sent, thus reducing the communication load and improving efficiency \cite{9955312}. Semantic communication is particularly beneficial in resource-constrained environments or systems requiring context-aware communication, such as edge computing and artificial intelligence (AI)-driven applications \cite{9398576,11130653}. By minimizing redundant information and prioritizing critical data, semantic communication reduces energy consumption and enhances overall system performance \cite{10550151}. In indoor scenarios, where fast and accurate data interpretation is essential, the combination of VLC and semantic communication presents a promising solution for future communication networks.

\subsection{Related Work}
A number of studies have examined the individual benefits of VLC and semantic communication, yet only a limited number of works have investigated the potential of combining these two technologies.

\subsubsection{VLC}
Research into VLC has confirmed its potential for high-speed, secure communication across diverse applications, from indoor wireless networks to vehicular systems \cite{9241073}. 
A significant body of work has focused on characterizing the unique properties of VLC channels. Unlike traditional RF, VLC links are distinguished by their immunity to small-scale fading and negligible Doppler effects \cite{8360928}. 
Modeling efforts have produced sophisticated representations that account for environmental factors such as obstacle geometry \cite{8675971} and challenging conditions like dust scattering in underground mines \cite{9222015}. 
Theoretical work has also advanced our understanding by establishing tight upper bounds on optical wireless channel capacity using novel approximation methods \cite{7323840}.

Security in VLC has been another critical research avenue. Studies have analyzed and enhanced this, for instance, by evaluating the secrecy performance of dual-hop hybrid free space optics (FSO)-VLC systems against eavesdroppers \cite{9888788} and by developing secure coding schemes using polar codes to ensure both physical-layer security and transmission reliability without compromising illumination functionality \cite{8463605}.

To support multi-user scenarios, various multiple access techniques have been adapted for multiple-input multiple-output (MIMO)-VLC systems, considering practical aspects like dynamic-range constraints and dimming control \cite{7562531,10562043}.
Initial explorations centered on orthogonal schemes like orthogonal frequency division multiple access (OFDMA), valued for its high spectrum utilization efficiency \cite{9137669} and its dual utility in both communication and indoor positioning \cite{7859314}. 
Spatial multiplexing techniques such as space division multiple access (SDMA) were also shown to enhance performance by leveraging the spatial distribution of users, proving superior to conventional OFDMA in certain contexts \cite{9141348}.
More recently, non-orthogonal multiple access (NOMA) has been investigated to improve sum rates through sophisticated power allocation \cite{8233180}, though its performance characteristics in VLC differ notably from those in RF systems \cite{10195155}.
Building on this, rate-splitting multiple access (RSMA) has emerged as a highly flexible framework \cite{7470942,9831440,10038476}. 
Studies have demonstrated RSMA's superiority over NOMA and SDMA in terms of weighted sum rate \cite{mao2018rate,9226406}, analyzed its contributions to spectral and energy efficiency under various channel conditions \cite{8907421,8491100,10499828}, explored its use in networks with simultaneous lightwave information and power transfer (SLIPT) \cite{10552701}, and optimized its energy efficiency in broadcast systems under practical quality-of-service (QoS) constraints \cite{9693949}.

\subsubsection{Semantic Communication}
The paradigm of semantic communication has gained considerable traction for its potential to drastically improve data transmission efficiency \cite{11006980}. To facilitate multi-user semantic systems, several multiple access strategies have been investigated, which can be broadly grouped by the underlying technique.

In orthogonal multiple access (OMA)-aided systems, resources are optimized to maximize spectral efficiency through channel assignment \cite{yan2022resource} or to enhance energy efficiency under latency and semantic fidelity constraints \cite{10689487, 9832831}. This has been extended to MEC systems where communication and computation are jointly managed \cite{10287133} and where semantic level selection can be refined to minimize energy consumption \cite{e26050394}.

For SDMA-aided semantic communication, precoding designs in multiple-antenna systems facilitate the transmission of distinct semantic information to multiple users. The performance of such systems has been evaluated using metrics like the age of incorrect information (AoII) \cite{10093842} and has been extended to consider multimodal semantic transmissions \cite{9921202}.

The integration of NOMA has been explored, particularly in heterogeneous scenarios where semantic-aware users coexist with conventional bit-interested users \cite{mu2022heterogeneous, 10158994}. Furthermore, a unified framework has been presented to accommodate diverse datasets and data modalities within NOMA-aided semantic systems \cite{10225385}.

RSMA offers a particularly synergistic pairing with semantic communication. Its intrinsic structure allows common knowledge for all users to be encoded into the common message, while user-specific semantic data is carried in private messages \cite{10032275,10615621}. This architecture has been leveraged to formulate maximization problems for the weighted sum of semantic information rates \cite{10287956} and to optimize the quality of experience (QoE) for specific tasks like image transmission \cite{10303275}.

Lastly, mode division multiple access (MDMA)-aided semantic communication enables information recovery for different users by analyzing resources in a high-dimensional semantic space, demonstrating superior performance over traditional FDMA and NOMA \cite{zhang2023model}. This approach has been applied to complex data types like 3D point clouds \cite{10679082} and developed using deep learning-based methods \cite{10288558}.

\subsubsection{VLC and Semantic Communication}
Despite the individual promise of VLC and semantic communication, their direct integration is a nascent but promising research area. The few pioneering studies have yielded encouraging results in terms of transmission efficiency and robustness. For instance, to address the inherent bandwidth limitations of VLC, image compression techniques specifically tailored for semantic transmission have been introduced, enhancing deployment feasibility by leveraging the high fidelity of semantic systems \cite{10636681}. 
Another approach has focused on optimizing both transmission efficiency and link stability for indoor image transmission, successfully extending the communication coverage of VLC and ensuring robustness even in high-speed scenarios, all without requiring hardware modifications \cite{Chen:24}.

\subsection{Motivation and Contributions}

\begin{table*}[ht]
\centering
\caption{Comparison with Closest Prior Art}
\begin{tabular}{|c|c|c|c|c|}
    \hline
    \textbf{Work} & \textbf{System Model} & \textbf{Communication Paradigm} & \textbf{Multiple Access} & \textbf{Semantic Computational Cost} \\
    \hline
    \cite{9693949} & VLC & Conventional & RSMA & None \\
    \cite{10636681,Chen:24} & VLC & Semantic & None & Ignore \\
    \cite{10032275,10287956} & RF & Semantic & RSMA & Consider as Smooth Convex Function \\
    \textbf{Our Work} & VLC & Semantic & RSMA & Consider as Segmented Linear Function \\
    \hline
\end{tabular}
\label{tb.rw}
\end{table*}

While prior studies \cite{10636681,Chen:24} have explored image transmission within VLC-based semantic communication systems, they have not integrated multiple access techniques, leaving the critical issue of energy efficiency largely unaddressed. To bridge this gap, this paper proposes an energy-efficient design for a RSMA-aided probabilistic semantic communication (PSCom) system over VLC networks. Our objective is to maximize the system's energy efficiency by jointly optimizing communication resources and the associated computation overhead. Notably, many existing works on semantic communication overlook the computational costs of semantic compression, which can lead to incomplete performance comparisons with conventional systems. The novelty of this work, therefore, lies not in the individual technologies of VLC, semantic communication, or RSMA, but in their synergistic integration to solve this previously unaddressed problem. By jointly considering the semantic computation overhead alongside communication overhead in a resource-constrained VLC network, we present a more realistic and comprehensive framework. To clearly delineate this research gap and situate our work, we provide a direct comparison with the closest prior art in Table ~\ref{tb.rw}. As the table illustrates, our work is the first, to our knowledge, to synergistically integrate VLC, RSMA, and semantic communication framework. The main contributions of this paper are summarized as follows:
\begin{itemize}
    \item We propose a framework for a VLC-based PSCom network where RSMA is specifically tailored to manage the crucial task of knowledge base synchronization alongside private data transmission. The common message in RSMA is uniquely used to broadcast updates to the shared probabilistic graphs, which is essential for the PSCom model.
    \item We formulate a comprehensive energy efficiency problem that integrates the semantic compression ratio with its associated computational power cost. To enable this, we derive an effective rate expression for PSCom that uniquely accounts for the overhead from semantic information extraction. The problem involves the joint optimization of transmit beamforming, direct current (DC) bias, rate allocation, and the semantic compression ratio. This creates a multi-layered optimization framework that captures the fundamental trade-offs between communication and computation in a resource-constrained VLC system.
    \item To address this problem, an alternating optimization algorithm is proposed, which iteratively tackles the transmit beamforming design and rate allocation subproblem, as well as the semantic compression ratio and DC bias design subproblem. The successive convex approximation (SCA) method is utilized to solve the transmit beamforming design and rate allocation subproblem efficiently. Subsequently, with the derived optimal DC bias, the Dinkelbach method is employed to address the semantic compression ratio and DC bias design subproblem. Numerical results demonstrate the effectiveness and superiority of the proposed algorithm.
\end{itemize}

The rest of this paper is organized as follows. Section~\ref{Sec:smpf} presents the VLC-based PSCom system model and formulates the energy efficiency maximization problem, considering both communication and computation costs. Section~\ref{Sec:ad} proposes an alternating optimization algorithm based on SCA and the Dinkelbach method to solve this problem. Section~\ref{Sec:sr} validates our proposed algorithm through simulations against several benchmarks, and Section~\ref{Sec:c} concludes the paper.

The main notations used in the paper are summarized in Table \ref{tb:mn}.


\begin{table}[ht]
\renewcommand\arraystretch{1.15} 
\centering
\caption{List of Main Notations}
\begin{tabular}{|c||l|} 
    \hline
    \textbf{Notation} & \textbf{Description} \\
    \hline \hline
    \multicolumn{2}{|c|}{\textbf{System Parameters}} \\
    \hline
    $\mathcal{N}$ & Set of LEDs, $\mathcal{N} = \{1, \dots, N\}$ \\ \hline
    $\mathcal{K}$ & Set of users, $\mathcal{K} = \{1, \dots, K\}$ \\ \hline
    $h_{n,k}$ & Channel gain between LED $n$ and user $k$ \\ \hline
    $\mathbf{h}_k$ & Channel gain vector for user $k$ \\ \hline
    $P^{\max}$ & Maximum total power budget of the transmitter \\ \hline
    $I_\mathrm{L}, I_\mathrm{U}$ & Lower and upper bounds of the LED drive current \\ \hline
    $B_{\mathrm{DC}}$ & Direct current bias \\ \hline
    $\sigma_k^2$ & Noise power for user $k$ \\ \hline
    $R_k$ & Minimum effective rate demand for user $k$ \\ \hline
    \multicolumn{2}{|c|}{\textbf{PSCom Variables}} \\
    \hline
    $\rho_k$ & Semantic compression ratio for user $k$ \\ \hline
    $\boldsymbol{\rho}$ & Vector of semantic compression ratios $[\rho_1, \dots, \rho_K]^\mathrm{T}$ \\ \hline
    $\rho_k^{\min}$ & Minimum semantic compression ratio for user $k$ \\ \hline
    $Q_k(\rho_k)$ & Computational overhead function for user $k$ \\ \hline
    $P_k^\mathrm{comp}$ & Computation power for user $k$ \\ \hline
    $R_0$ & Minimum rate demand for knowledge updates \\ \hline
    \multicolumn{2}{|c|}{\textbf{RSMA Variables}} \\
    \hline
    $\mathbf{w}_0$ & Beamforming vector for the common message \\ \hline
    $\mathbf{w}_k$ & Beamforming vector for the private message $k$ \\ \hline
    $\mathbf{W}$ & Collection of all beamforming vectors $[\mathbf{w}_0; \dots; \mathbf{w}_K]$ \\ \hline
    $a_0$ & Common rate for knowledge updates \\ \hline
    $a_k$ & Common rate for user $k$'s data \\ \hline
    $\mathbf{a}$ & Vector of all common rates $[a_0, \dots, a_K]^\mathrm{T}$ \\ \hline
    $c_k$ & Achievable rate for common message at user $k$ \\ \hline
    $r_k$ & Achievable rate for private message at user $k$ \\ \hline
    $r_k^\mathrm{eff}$ & Effective semantic rate for user $k$ \\ \hline
    \multicolumn{2}{|c|}{\textbf{Algorithm \& Optimization Variables}} \\
    \hline
    $\alpha, \beta, \gamma$ & Auxiliary variables for SCA in Problem \eqref{eq.tbra2} \\ \hline
    $\delta_k^\mathrm{c}, \zeta_k^\mathrm{c}, \mu_k^\mathrm{c}$ & Auxiliary variables for common rate $c_k$ \\ \hline
    $\delta_k^\mathrm{r}, \zeta_k^\mathrm{r}, \mu_k^\mathrm{r}$ & Auxiliary variables for private rate $r_k$ \\ \hline
    $\lambda$ & Dinkelbach method parameter \\ \hline
\end{tabular}
\label{tb:mn}
\end{table}

\section{System Model and Problem Formulation}\label{Sec:smpf}

\begin{figure}[t]
    \centering
    \includegraphics[width=\linewidth]{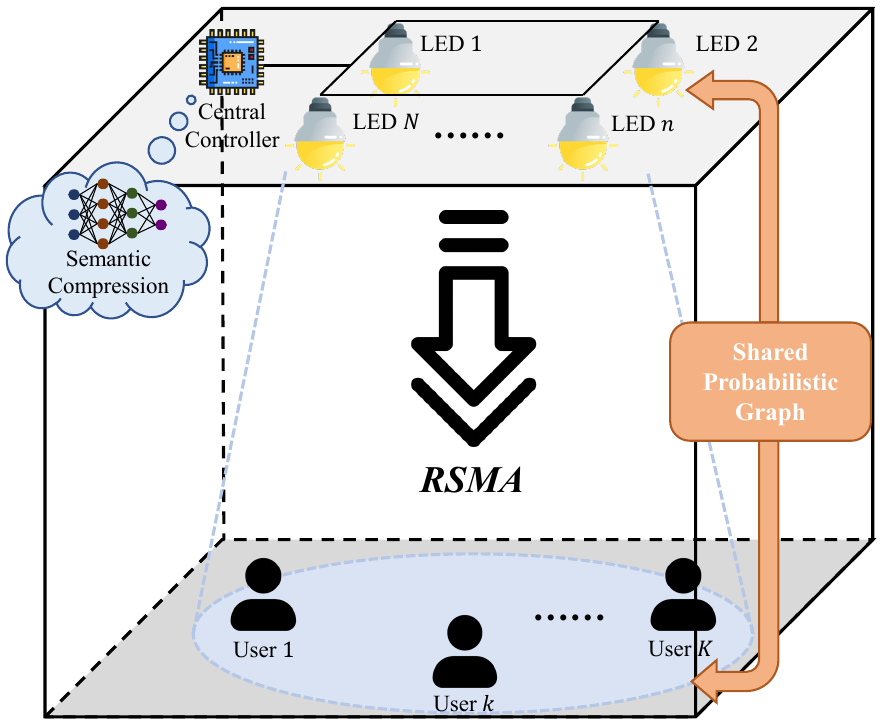}
    \caption{An indoor VLC-based PSCom network with multiple LEDs and users.}
    \label{fig.sm}
\end{figure}

Consider an indoor VLC-based PSCom network consisting of $N$ LEDs and $K$ users, as depicted in Fig.~\ref{fig.sm}. The set of LEDs is denoted by $\mathcal{N}$, while the set of users is denoted by $\mathcal{K}$. All LEDs are connected to a central controller to facilitate cooperative transmission, and each user is equipped with a single photo detector (PD) to receive the optical signals. At the transmitter side, the original data is first compressed using the PSCom technique before being transmitted to the users via LEDs using RSMA. At the receiver side, the users decode the compressed semantic information and perform semantic inference to recover the original data, aided by shared probabilistic graphs.

\subsection{PSCom Model}
In PSCom network, probabilistic graphs are adopted as the knowledge base between transceivers. Compared to conventional communication systems, PSCom exhibits three main properties:
\begin{itemize}
    \item \textbf{Semantic Representation:} Semantic information is compressed using the probabilistic graph, which reduces redundant information. Here, redundant information refers to the messages that can be inferred through the probabilistic graph.
    \item \textbf{Knowledge Transmission:} As both the transmitter and receiver share a common knowledge base, the transmission of both knowledge and information must be considered within PSCom. This necessitates the use of multiple access techniques to facilitate knowledge transmission.
    \item \textbf{Effective Rate Calculation:} In PSCom, each received bit can convey more information than a conventional bit due to the information compression process at the transmitter side. Consequently, the formulation for effective rate calculation differs in PSCom.
\end{itemize}

\begin{figure*}[t]
    \centering
    \includegraphics[width=\linewidth]{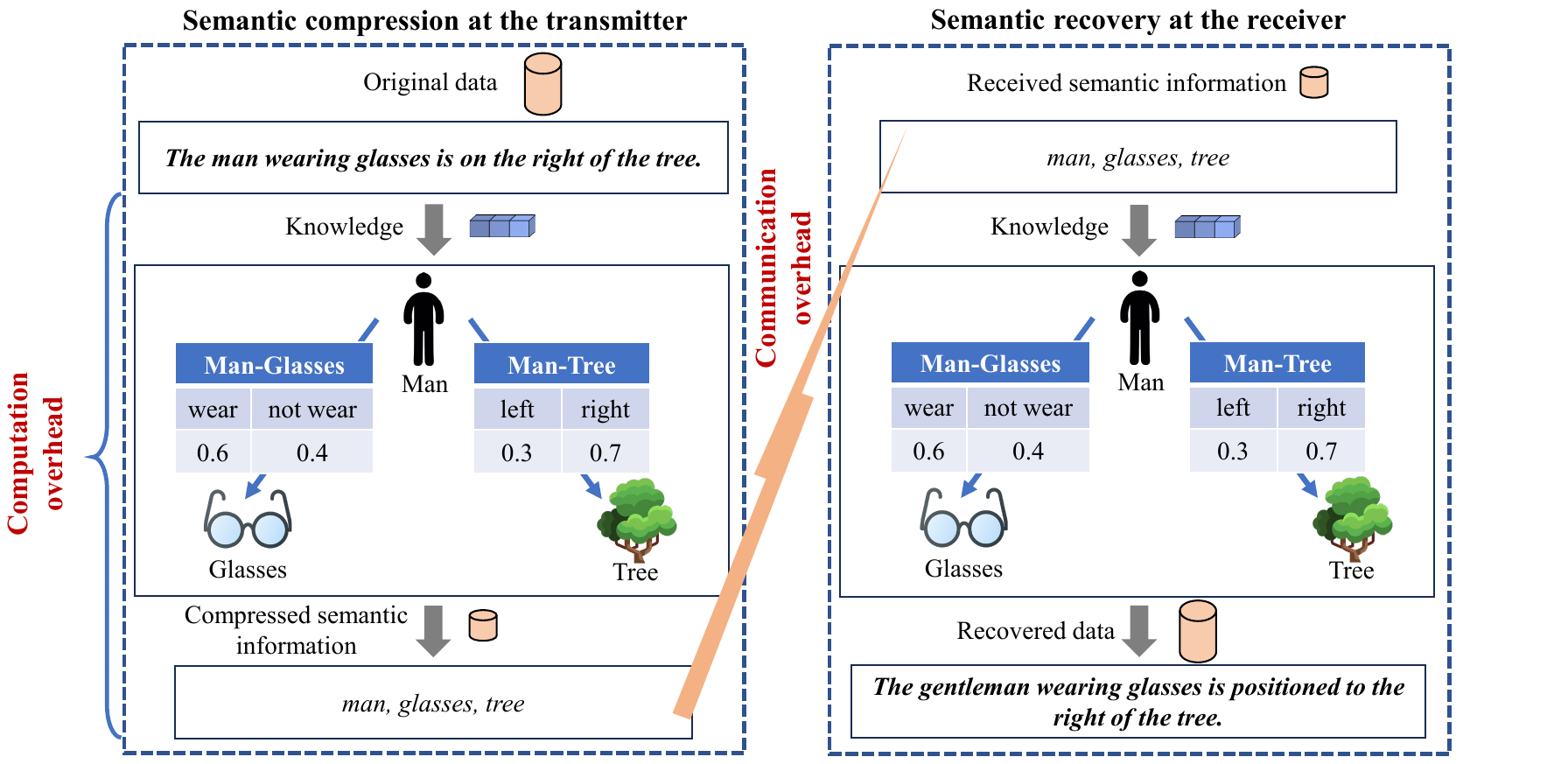}
    \caption{An example of the semantic compression and semantic recovery procedures with the aid of knowledge characterized by probabilistic graph.}
    \label{fig.PSCom}
\end{figure*}

Fig.~\ref{fig.PSCom} illustrates the operational mechanism of PSCom. At the transmitter side, original data undergoes semantic compression utilizing the common knowledge base, represented as probabilistic graphs within the PSCom system. The probabilistic graph functions as a Bayesian network that captures the probabilistic dependencies among various entities. These relationships are derived from a large number of samples and can be updated throughout the information transmission process to maintain knowledge freshness. Notably, multi-layer semantic compression can be achieved using multi-dimensional conditional probabilities, potentially allowing for greater compression of information with increased computational effort. At the receiver side, after receiving the compressed semantic information, semantic inference is performed to recover the data with the assistance of the shared probabilistic graph. The recovered data retains the same semantic meaning as the original data. The PSCom framework is applicable to any data that can be semantically represented by graphs. This includes various modalities, including text, images, and videos. Furthermore, it maintains a distinct separation between semantic source coding and digital channel transmission, compatible with existing digital communication systems. Technical details regarding PSCom can be found in \cite{ZHAO2024107055}.

\subsubsection{Knowledge Update}
As previously indicated, the PSCom framework necessitates regular knowledge updates. The minimum knowledge update rate, denoted by $R_0$, is a pivotal QoS parameter that dictates the currency of the shared knowledge base. The probabilistic graph that serves as the foundation of the PSCom model is not static; it evolves as new data and semantic concepts are observed, and its accuracy can degrade over time. Therefore, $R_0$ signifies the minimum rate necessary to maintain the synchronization and relevance of this knowledge base. This minimum rate can be formally characterized by the age of information (AoI) of the graph.

In practical scenarios, the value of $R_0$ is not a fixed constant but is highly dependent on the application's context. For instance, in a highly dynamic environment where semantic relationships change rapidly, such as in a live video feed or a fast-evolving conversation, a higher $R_0$ is necessary to maintain a low AoI. Conversely, for tasks with relatively static data sources, such as transmitting a document or an archived image, $R_0$ can be set to a very low value after an initial synchronization.

The setting of $R_0$ introduces a fundamental trade-off between communication and computation energy efficiency, which is central to this work. A higher $R_0$ ensures that the shared probabilistic graph has a low AoI, making it more accurate and relevant. This enhanced knowledge base allows the PSCom encoder to achieve a better semantic compression ratio for a given computational effort, thereby improving the computation energy efficiency. However, this benefit comes at the cost of increased communication power consumption, as more resources must be dedicated to broadcasting the common knowledge update message.

\subsubsection{Computational Overhead}
To characterize the transmission rate of semantic information, we introduce the concept of the semantic compression ratio. The semantic compression ratio for user $k$ is defined as
\begin{equation}\label{eq.rho}
    \rho_k=\frac{\mathrm{size}(\mathcal{C}_k)}{\mathrm{size}(\mathcal{D}_k)},
\end{equation}
where $\mathcal{D}_k$ represents the original data of user $k$, $\mathcal{C}_k$ denotes the compressed semantic information of user $k$, and the function $\mathrm{size}(\cdot)$ quantifies the data size.

As shown in Fig.~\ref{fig.PSCom}, the PSCom system encompasses both communication and computation procedures. The semantic compression operation inherently introduces additional computational overhead. According to Equation (20) in \cite{ZHAO2024107055}, the computational overhead exhibits a segmented linear relationship with respect to the semantic compression ratio. Assuming that the semantic compression ratio has $M$ levels, the computation overhead function for user $k$ can be mathematically modeled as
\begin{equation}\label{eq.comp}
    Q_k\left(\rho_k\right)=\max_{m=1,\cdots,M}\left(A_{km}\rho_k+B_{km}\right),\forall k\in\mathcal{K}, 
\end{equation}
where $A_{km}<0$ denotes the slope of segment $m$, and $B_{km}$ is a positive constant. These parameters are intrinsically linked to the structure of the probabilistic graph.

\begin{figure}[t]
    \centering
    \includegraphics[width=\linewidth]{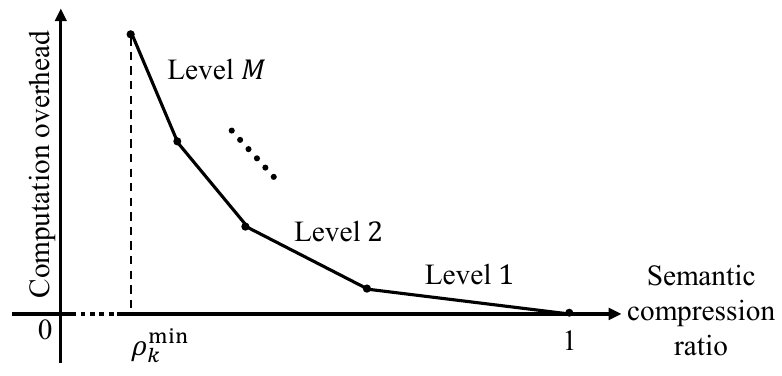}
    \caption{An illustration of the segmented linear relationship between semantic compression ratio and computation overhead in the PSCom system.}
    \label{fig.Comp}
\end{figure}

It is important to note that this piecewise linear model is an analytically validated characteristic adopted from the foundational PSCom framework presented in \cite{ZHAO2024107055}. The model's structure arises from the different levels of logical inference performed on the probabilistic graph to achieve compression. A lower compression ratio requires deeper, more complex inference across multiple probabilistic links, which incurs a greater computational cost, explaining the increasingly steep slope for higher compression levels shown in Fig.~\ref{fig.Comp}. The parameters $A_{km}$ and $B_{km}$ are system-specific constants that are determined by profiling the computational cost of a given trained PSCom model at each inference level and fitting a linear function to the resulting data points. Notably, for $\rho_k=1$, no semantic compression occurs, and we have $Q_k\left(1\right)=0$ for the computation overhead. Additionally, each user is subject to a semantic compression limit, denoted by $\rho_k^{\min}$, which represents the lower bound of the semantic compression ratio.

\subsection{VLC Channel Model}

\begin{figure}[t]
    \centering
    \includegraphics[width=0.7\linewidth]{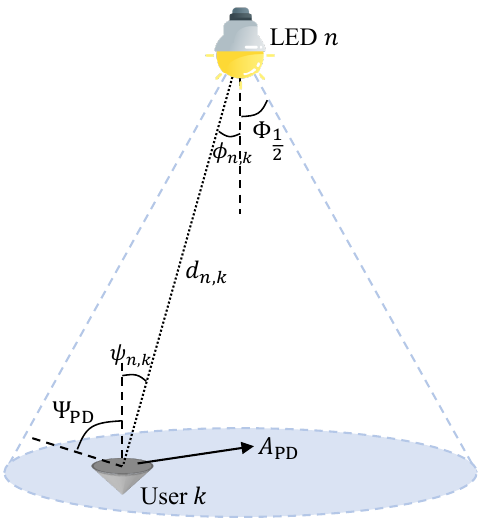}
    \caption{An illustration of the Lambertian radiation model.}
    \label{fig.lambert}
\end{figure}

In an indoor VLC environment, the optical signal propagates from the LED to the user via multiple paths, which include a direct line-of-sight (LoS) path and multiple non-line-of-sight (NLoS) paths caused by reflections. While NLoS components can introduce multipath effects, it is widely documented that the LoS link is dominant and typically carries over 90\% of the total received power when an uninterruptible path exists \cite{9693949,9829848}. Given that our work focuses on the system-level optimization of communication and computation resources rather than on fine-grained channel characterization, we adopt the common and tractable approach of considering only the dominant LoS component. This allows us to focus on the complexities of the resource allocation problem itself.
For the channel model, we assume that the receiving surfaces of all user PDs are oriented parallel to the floor, pointing vertically towards the ceiling. This standard assumption represents a typical indoor usage scenario and ensures the analytical tractability of the channel gain.
Employing the Lambertian radiation model illustrated in Fig.~\ref{fig.lambert}, the channel gain between LED $n$ and user $k$ can be expressed as
\begin{equation}
    h_{n,k}=\frac{(m+1) A_{\mathrm{PD}} \cos^m\left(\phi_{n,k}\right) \cos \left(\psi_{n,k}\right) G g\left(\psi_{n,k}\right)}{2 \pi d_{n,k}^2}.
\end{equation}
Here, $m=-\ln 2/\ln\left(\cos\Phi_{\frac{1}{2}}\right)$ is the Lambert index, with $\Phi_{\frac{1}{2}}$ representing the semi-angle of the LED. $A_{\mathrm{PD}}$ denotes the physical area of the PD, while $\phi_{n,k}$ and $\psi_{n,k}$ are the angles of irradiance and incidence, respectively. Moreover, $G$ denotes the optical filter gain, $d_{n,k}$ is the distance between LED $n$ and user $k$, and $g\left(\psi_{n,k}\right)$ denotes the optical concentrator gain, which can be expressed as
\begin{equation}
    g\left(\psi_{n,k}\right)=
    \begin{cases}
        \frac{\kappa^2}{\sin^2\left(\Psi_{\mathrm{PD}}\right)}, & 0 \leq \psi_{n,k} \leq \Psi_{\mathrm{PD}},\\
        0, & \psi_{n,k} \geq \Psi_{\mathrm{PD}},
    \end{cases}
\end{equation}
where $\kappa$ represents the refractive index, and $\Psi_{\mathrm{PD}}$ represents the field of view (FoV) of the PD.

\subsection{RSMA for VLC-Based PSCom}
As previously discussed, knowledge updates are required in the PSCom system to enhance the performance of semantic compression. We denote the knowledge intended for all users by the data stream $x_0$, and the semantic information specific to user $k$ by $x_k$. To facilitate the simultaneous transmission of both knowledge and semantic information, we adopt the rate splitting-aided unicast and multicast (RSUM) scheme \cite{8846706}. In this scheme, knowledge and a portion of the semantic data stream intended for each user are encoded into common message, while the remaining portion is encoded into private message, as shown in Fig.~\ref{fig.rsum}. Specifically, the data stream $x_k$ is divided into a common part $x_{k1}$ and a private part $x_{k2}$. The common message $s_0$ encodes the data streams $\{x_0, x_{11}, \ldots, x_{K1}\}$, while the private message $s_k$ encodes the data stream $x_{k2}$ \cite{7555358}.

\begin{figure}[t]
    \centering
    \includegraphics[width=0.9\linewidth]{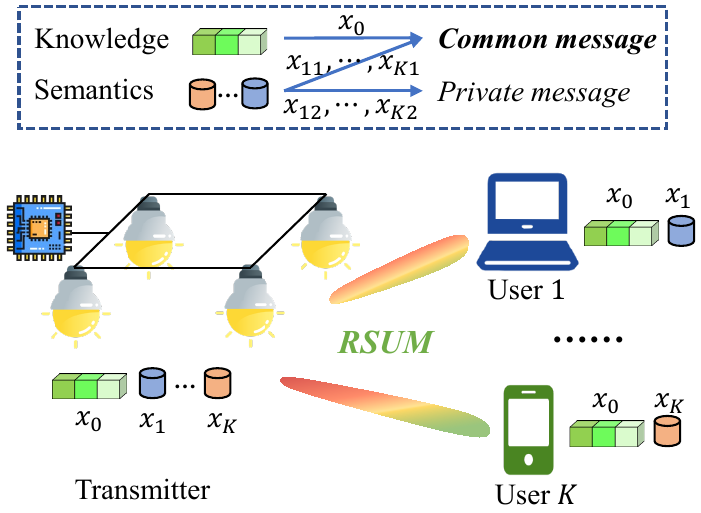}
    \caption{The mechanism of the RSUM scheme.}
    \label{fig.rsum}
\end{figure}

Without loss of generality, we normalize $s_i$ to the range $\left[-1,1\right]$ with zero mean, where $i=\left\{0,1,\cdots,K\right\}$. Given the characteristics of LEDs, the electrical signal must be non-negative and real-valued; thus, a DC bias $B_{\mathrm{DC}}$ is added to the transmitted signal. The transmit signal of the LEDs can be written as
\begin{equation}
    \mathbf{x}=\mathbf{w}_0 s_0 + \sum_{l=1}^K \mathbf{w}_l s_l + \mathbf{B}_{\mathrm{DC}},
\end{equation}
where $\mathbf{w}_0=\left[w_{0,1},\cdots,w_{0,N}\right]^\mathrm{T}$ and $\mathbf{w}_l=\left[w_{l,1},\cdots,w_{l,N}\right]^\mathrm{T}$ are the real-valued beamforming vectors, and $\mathbf{B}_{\mathrm{DC}}=\left(B_{\mathrm{DC}}\right)_{N\times 1}$ denotes the DC bias. To ensure that each LED operates within its dynamic range, the amplitude of $x_i$ must remain bounded between $I_{\mathrm{L}}$ and $I_{\mathrm{U}}$ to prevent clipping, which can be mathematically represented as
\begin{equation}
    \sum_{l=0}^K \left|w_{l,i}\right| \leq \min\left\{B_{\mathrm{DC}}-I_\mathrm{L},I_\mathrm{U}-B_{\mathrm{DC}}\right\}, \forall i\in\mathcal{N},
\end{equation}
where $I_\mathrm{L}$ and $I_\mathrm{U}$ represent the lower and upper bounds of the LED drive current with the linear region.

In the considered VLC scenario, it is assumed that the LEDs utilize intensity modulation, while the users employ a direct detection technique. The transmit signal $\mathbf{x}$ induces each LED to emit an optical signal, the intensity of which is directly proportional to the magnitude of the signal $\mathbf{x}$. At the receiving end, the optical signal $\mathbf{x}$ is detected and subsequently converted into an electric signal by the PD of each user. For user $k$, the received signal can be given as
\begin{equation}
    y_k=\mathbf{h}_k^\mathrm{T}\left(\mathbf{w}_0 s_0 + \sum_{l=1}^K \mathbf{w}_l s_l + \mathbf{B}_{\mathrm{DC}}\right)+n_k,
\end{equation}
where $\mathbf{h}_k=\left[h_{1,k},\cdots,h_{N,k}\right]^\mathrm{T}$ is the channel gain vector, and $n_k$ represents the additive white Gaussian noise (AWGN) with zero mean and variance $\sigma^2_k$.

Since the DC bias carries no information, it is removed at the receiver using alternating current (AC) coupling. In RSUM scheme, the common message $s_0$ is decoded first, treating the other private messages as interference. For user $k$, the rate for decoding the common message $s_0$ can be given by
\begin{equation}\label{eq.commonrate}
    c_k=\log_2\left(1+\frac{2}{\pi e}\frac{\left|\mathbf{h}_k^\mathrm{T}\mathbf{w}_0\right|^2}{\sum_{l=1}^{K}\left|\mathbf{h}_k^\mathrm{T}\mathbf{w}_l\right|^2+\sigma_k^2}\right),
\end{equation}
where $2/\pi e$ is a typical coefficient for achievable rate in VLC system.

The common message comprises both the knowledge data stream and the individual data stream for each user. Let $a_0$ denote the rate of the knowledge data stream and $a_k$ represent the rate of the individual data stream for user $k$. To ensure that all users can successfully decode the common message, we have
\begin{equation}
    a_0+\sum_{l=1}^K a_l\leq c_k,\forall k \in \mathcal{K}.
\end{equation}

After subtracting the successfully decoded common message, the rate for decoding private message $s_k$ at user $k$ can be given by
\begin{equation}
    r_k=\log_2\left(1+\frac{2}{\pi e}\frac{\left|\mathbf{h}_k^\mathrm{T}\mathbf{w}_k\right|^2}{\sum_{l=1,l\neq k}^{K}\left|\mathbf{h}_k^\mathrm{T}\mathbf{w}_l\right|^2+\sigma_k^2}\right).
\end{equation}

The overall rate for user $k$ encompasses both the common message part $a_k$ and the private message part $r_k$. Recalling \eqref{eq.rho}, the effective rate of user $k$ can be given by 
\begin{equation}
    r_k^{\text{eff}}=\frac{a_k+r_k}{\rho_k},
\end{equation}
where the coefficient $1/\rho_k$ accounts for the fact that one bit of obtained semantic information can convey more than one bit (i.e., $1/\rho_k$ bits) of the original information.

\subsection{Power Consumption Model}
In the considered VLC-based PSCom system, the total power consumed at the transmitter comprises both the computation power required for semantic compression and the communication power associated with information transmission.

\subsubsection{Computation Power}
Based on \eqref{eq.comp}, the computation power for semantic compression of user $k$ can be written as
\begin{equation}
    P_k^{\text{comp}}=\eta\max_{m=1,\cdots,M}\left(A_{km}\rho_k+B_{km}\right),
\end{equation}
where $\eta$ is the power coefficient dependent on the chip architecture. Consequently, the total computation power at the transmitter is $P^{\text{comp}}=\sum_{l=1}^K P_l^{\text{comp}}$.

\subsubsection{Communication Power}
In the VLC system, the communication power consists of two components: the AC part and the DC part. The power associated with the AC component can be expressed as
\begin{equation}
    P^{\text{AC}}=\sum_{l=0}^K\left\|\mathbf{w}_l\right\|^2.
\end{equation}
The power consumed by the DC component is given by
\begin{equation}
    P^{\text{DC}}=N U_{\text{LED}} B_{\mathrm{DC}} + P^{\text{cir}},
\end{equation}
where $U_{\text{LED}}$ denotes the forward voltage of the LEDs, and $P^{\text{cir}}$ represents the circuit power which is a constant. Thus, the total communication power is $P^{\text{comm}}=P^{\text{AC}}+P^{\text{DC}}$.

\subsection{Problem Formulation}
Based on the preceding models, the central objective of this work is to maximize the overall energy efficiency of the VLC-based PSCom system. This is achieved by jointly optimizing the transmit beamforming matrices, the DC bias, the common rate allocation, and the per-user semantic compression ratios. We define energy efficiency as the ratio of the sum effective user rates to the total power consumption. This leads to the following optimization problem:
\begin{subequations}\label{eq.pf}
    \begin{align}
		\max_{\mathbf{W},B_{\mathrm{DC}},\mathbf{a},\boldsymbol\rho} & \frac{\sum_{k=1}^K r_k^{\text{eff}}}{P^{\text{comp}}+P^{\text{comm}}}\tag{\theequation}\\
		\mathrm{s.t.}\hspace{1.2em}
        & P^{\text{comp}}+P^{\text{comm}}\leq P^{\max},\label{cons.power}\\
        & \sum_{l=0}^K \left|w_{l,i}\right| \leq \min\left\{B_{\mathrm{DC}}-I_\mathrm{L},I_\mathrm{U}-B_{\mathrm{DC}}\right\},\notag\\
        &\hspace{11em} \forall i\in\mathcal{N},\label{cons.LED}\\
        & a_0+\sum_{l=1}^K a_l\leq c_k,\forall k \in \mathcal{K},\label{cons.RSMA}\\
        & a_0\geq R_0,\label{cons.KU}\\
		& r_k^{\text{eff}}\geq R_k,\forall k \in \mathcal{K},\label{cons.minrate}\\
        & a_k\geq 0,\forall k \in \mathcal{K},\label{cons.a0}\\
        & \rho_k^{\min}\leq\rho_k\leq 1,\forall k \in \mathcal{K},\label{cons.rho}
    \end{align}
\end{subequations}
where $\mathbf{W}=[\mathbf{w}_0;\mathbf{w}_1;\cdots;\mathbf{w}_K]$, $\mathbf{a}=\left[a_0,a_1,\cdots,a_K\right]^\mathrm{T}$, and $\boldsymbol\rho=[\rho_1,\cdots,\rho_K]^\mathrm{T}$. Here, $P^{\max}$ denotes the maximum power of the transmitter, $R_0$ represents the minimum rate demand for knowledge updates in PSCom, and $R_k$ is the minimum rate demand for user $k$.
The constraints ensure proper system operation. Specifically, \eqref{cons.power} enforces the total transmitter power budget, $P^{\max}$, while \eqref{cons.LED} guarantees that the LED signals remain within their linear dynamic range. The QoS requirements are captured by a set of rate constraints: \eqref{cons.RSMA} governs common message decoding for the RSMA scheme, \eqref{cons.KU} ensures a minimum rate $R_0$ for knowledge updates, and \eqref{cons.minrate} maintains a minimum effective rate $R_k$ for each user. Finally, \eqref{cons.a0} and \eqref{cons.rho} define the feasible ranges for the rate allocation and semantic compression ratio variables.

The optimization problem in \eqref{eq.pf} is a non-convex fractional program due to the coupled variables and the objective function's structure, making it computationally intractable to find the global optimum directly. Therefore, in the subsequent section, we propose an efficient iterative algorithm to find a high-quality suboptimal solution.



\section{Algorithm Design}\label{Sec:ad}
To address problem \eqref{eq.pf}, an alternating optimization algorithm is proposed. This algorithm iteratively solves two subproblems, the transmit beamforming design and rate allocation subproblem, and the semantic compression ratio and DC bias design subproblem, until convergence is achieved.

\subsection{Transmit Beamforming Design and Rate Allocation}
With given semantic compression ratio $\boldsymbol\rho$ and DC bias $B_{\mathrm{DC}}$ in problem \eqref{eq.pf}, the transmit beamforming design and rate allocation subproblem can be expressed as
\begin{subequations}\label{eq.tbra}
    \begin{align}
		\max_{\mathbf{W},\mathbf{a}}\quad & \frac{\sum_{k=1}^K \left(a_k+r_k\right)/\rho_k}{\sum_{l=0}^K\left\|\mathbf{w}_l\right\|^2 +\epsilon_1}\tag{\theequation}\\
		\mathrm{s.t.}\hspace{1.1em}
        & \sum_{l=0}^K\left\|\mathbf{w}_l\right\|^2\leq P^{\max}-\epsilon_1,\label{eq.tbra.c1}\\
        & \sum_{l=0}^K \left|w_{l,i}\right| \leq \epsilon_2,\forall i\in\mathcal{N},\label{eq.tbra.c2}\\
        & \eqref{cons.RSMA}-\eqref{cons.a0},\notag
    \end{align}
\end{subequations}
where $\epsilon_1=\sum_{l=1}^K P_l^{\text{comp}}+N U_{\text{LED}} B_{\mathrm{DC}} + P^{\text{cir}}$ and $\epsilon_2=\min\left\{B_{\mathrm{DC}}-I_\mathrm{L},I_\mathrm{U}-B_{\mathrm{DC}}\right\}$ are constants. Problem \eqref{eq.tbra} remains challenging due to its non-convex fractional nature. To address this, we employ the SCA method as in \cite{9693949}.

To handle the non-convexity of the objective function in problem \eqref{eq.tbra}, we introduce three auxiliary scalar variables, $\alpha,\beta,\gamma$. Then, problem \eqref{eq.tbra} can be reformulated as
\begin{subequations}\label{eq.tbra2}
    \begin{align}
		\max_{\mathbf{W},\mathbf{a},\alpha,\beta,\gamma}\quad & \gamma\tag{\theequation}\\
		\mathrm{s.t.}\hspace{2em}
        & \frac{\alpha^2}{\beta} \geq \gamma,\label{eq.tbra2.c1}\\
        & \sum_{k=1}^K \frac{a_k+r_k}{\rho_k} \geq \alpha^2,\label{eq.tbra2.c2}\\
        & \sum_{l=0}^K\left\|\mathbf{w}_l\right\|^2 +\epsilon_1 \leq \beta,\label{eq.tbra2.c3}\\
        & \eqref{eq.tbra.c1},\eqref{eq.tbra.c2},\eqref{cons.RSMA}-\eqref{cons.a0}.\notag
    \end{align}
\end{subequations}

In problem \eqref{eq.tbra2}, constraints \eqref{eq.tbra2.c1}, \eqref{eq.tbra2.c2}, \eqref{cons.RSMA}, and \eqref{cons.minrate} are still non-convex. Among constraints \eqref{eq.tbra2.c2}, \eqref{cons.RSMA}, and \eqref{cons.minrate}, the non-convexities are caused by the complicated expressions of $c_k$ and $r_k$.

To deal with the non-convexity of constraint \eqref{cons.RSMA}, we introduce slack variables $\delta^\mathrm{c}_k$ to represent $c_k$. Then, constraint \eqref{cons.RSMA} can be transformed into
\begin{align}
    & a_0+\sum_{l=1}^K a_l \leq \delta^\mathrm{c}_k, &\forall k \in \mathcal{K},\label{cons.RSMA.1}\\
    & c_k \geq \delta^\mathrm{c}_k, &\forall k \in \mathcal{K}.\label{cons.RSMA.2}
\end{align}
Among these two constraints, \eqref{cons.RSMA.1} is convex, but \eqref{cons.RSMA.2} is non-convex. Recall \eqref{eq.commonrate}, we further introduce auxiliary variables $\zeta^\mathrm{c}_k,\mu^\mathrm{c}_k$ to reformulate \eqref{cons.RSMA.2} as
\begin{align}
    & \zeta^\mathrm{c}_k \geq 2^{\delta^\mathrm{c}_k}, &\forall k \in \mathcal{K},\label{cons.RSMA.2.1}\\
    & \frac{\left|\mathbf{h}_k^\mathrm{T}\mathbf{w}_0\right|^2}{\mu^\mathrm{c}_k} \geq \frac{\pi e}{2}\left(\zeta^\mathrm{c}_k-1\right), &\forall k \in \mathcal{K},\label{cons.RSMA.2.2}\\
    & \mu^\mathrm{c}_k \geq \sum_{l=1}^{K}\left|\mathbf{h}_k^\mathrm{T}\mathbf{w}_l\right|^2+\sigma_k^2, &\forall k \in \mathcal{K}.\label{cons.RSMA.2.3}
\end{align}
With the above transformations, constraint \eqref{cons.RSMA} can be equivalently replaced by the combination of \eqref{cons.RSMA.1}, \eqref{cons.RSMA.2.1}-\eqref{cons.RSMA.2.3}.

Similarly, for constraints \eqref{eq.tbra2.c2} and \eqref{cons.minrate}, which are related to $r_k$, we introduce auxiliary variables $\delta^\mathrm{r}_k,\zeta^\mathrm{r}_k,\mu^\mathrm{r}_k$ to transform them into
\begin{align}
    & \sum_{k=1}^K \frac{a_k+\delta^\mathrm{r}_k}{\rho_k} \geq \alpha^2,\label{cons.rk.1}\\
    & \frac{a_k+\delta^\mathrm{r}_k}{\rho_k} \geq R_k, &\forall k \in \mathcal{K},\label{cons.rk.2}\\
    & \zeta^\mathrm{r}_k \geq 2^{\delta^\mathrm{r}_k}, &\forall k \in \mathcal{K},\label{cons.rk.3}\\
    & \frac{\left|\mathbf{h}_k^\mathrm{T}\mathbf{w}_k\right|^2}{\mu^\mathrm{r}_k} \geq \frac{\pi e}{2}\left(\zeta^\mathrm{r}_k-1\right), &\forall k \in \mathcal{K},\label{cons.rk.4}\\
    & \mu^\mathrm{r}_k \geq \sum_{l=1,l\neq k}^{K}\left|\mathbf{h}_k^\mathrm{T}\mathbf{w}_l\right|^2+\sigma_k^2, &\forall k \in \mathcal{K}.\label{cons.rk.5}
\end{align}

Among the above reformulations, constraints \eqref{eq.tbra2.c1}, \eqref{cons.RSMA.2.2}, and \eqref{cons.rk.4} remain non-convex. For these three constraints, we utilize first-order Taylor series to approximate them and employ SCA algorithm to obtain the asymptotically optimal solution.

The left-hand side of constraint \eqref{eq.tbra2.c1} can be linearly approximated by
\begin{equation}\label{fots.1}
    \frac{\alpha^2}{\beta} \approx \frac{2\alpha^{(i)}}{\beta^{(i)}}\alpha-\left(\frac{\alpha^{(i)}}{\beta^{(i)}}\right)^2\beta,
\end{equation}
where the superscript $(i)$ indicates the $i$-th iteration of the SCA algorithm. Note that \eqref{fots.1} is a lower-bound approximation.

Similarly, the left-hand side of constraint \eqref{cons.RSMA.2.2} can be approximated by
\begin{align}\label{fots.2}
    &\frac{\left|\mathbf{h}_k^\mathrm{T}\mathbf{w}_0\right|^2}{\mu^\mathrm{c}_k} \approx \frac{2\left(\mathbf{w}_0^{(i)}\right)^\mathrm{T}\mathbf{h}_k\mathbf{h}_k^\mathrm{T}\mathbf{w}_0}{\left(\mu^\mathrm{c}_k\right)^{(i)}}-\left(\frac{\mathbf{h}_k^\mathrm{T}\mathbf{w}_0^{(i)}}{\left(\mu^\mathrm{c}_k\right)^{(i)}}\right)^2\mu^\mathrm{c}_k,\notag\\
    &\hspace{15.5em} \forall k \in \mathcal{K}.
\end{align}
Regarding the left-hand side of constraint \eqref{cons.rk.4}, we have
\begin{align}\label{fots.3}
    &\frac{\left|\mathbf{h}_k^\mathrm{T}\mathbf{w}_k\right|^2}{\mu^\mathrm{r}_k} \approx \frac{2\left(\mathbf{w}_k^{(i)}\right)^\mathrm{T}\mathbf{h}_k\mathbf{h}_k^\mathrm{T}\mathbf{w}_k}{\left(\mu^\mathrm{r}_k\right)^{(i)}}-\left(\frac{\mathbf{h}_k^\mathrm{T}\mathbf{w}_k^{(i)}}{\left(\mu^\mathrm{r}_k\right)^{(i)}}\right)^2\mu^\mathrm{r}_k,\notag\\
    &\hspace{15.5em} \forall k \in \mathcal{K}.
\end{align}

Eventually, at the $i$-th iteration of the SCA algorithm, problem \eqref{eq.tbra} can be approximated by
\begin{subequations}\label{eq.tbra3}
    \begin{align}
		\max_{\mathbf{W},\mathbf{a},\alpha,\beta,\gamma,\boldsymbol\delta,\boldsymbol\zeta,\boldsymbol\mu}\quad & \gamma\tag{\theequation}\\
		\mathrm{s.t.}\hspace{3em}
        &\hspace{-2em} \frac{2\alpha^{(i)}}{\beta^{(i)}}\alpha-\left(\frac{\alpha^{(i)}}{\beta^{(i)}}\right)^2\beta \geq \gamma,\\
        &\hspace{-2em} \frac{2\left(\mathbf{w}_0^{(i)}\right)^\mathrm{T}\mathbf{h}_k\mathbf{h}_k^\mathrm{T}\mathbf{w}_0}{\left(\mu^\mathrm{c}_k\right)^{(i)}}-\left(\frac{\mathbf{h}_k^\mathrm{T}\mathbf{w}_0^{(i)}}{\left(\mu^\mathrm{c}_k\right)^{(i)}}\right)^2\mu^\mathrm{c}_k\notag\\
        & \hspace{1em} \geq \frac{\pi e}{2}\left(\zeta^\mathrm{c}_k-1\right),\forall k \in \mathcal{K},\\
        &\hspace{-2em} \frac{2\left(\mathbf{w}_k^{(i)}\right)^\mathrm{T}\mathbf{h}_k\mathbf{h}_k^\mathrm{T}\mathbf{w}_k}{\left(\mu^\mathrm{r}_k\right)^{(i)}}-\left(\frac{\mathbf{h}_k^\mathrm{T}\mathbf{w}_k^{(i)}}{\left(\mu^\mathrm{r}_k\right)^{(i)}}\right)^2\mu^\mathrm{r}_k\notag\\
        & \hspace{1em} \geq \frac{\pi e}{2}\left(\zeta^\mathrm{r}_k-1\right),\forall k \in \mathcal{K},\\
        &\hspace{-2em} \eqref{cons.KU},\eqref{cons.a0},\eqref{eq.tbra.c1},\eqref{eq.tbra.c2},\eqref{eq.tbra2.c3},\eqref{cons.RSMA.1},\notag\\
        &\hspace{-2em} \eqref{cons.RSMA.2.1},\eqref{cons.RSMA.2.3},\eqref{cons.rk.1},\eqref{cons.rk.2},\eqref{cons.rk.3},\eqref{cons.rk.5},\notag
    \end{align}
\end{subequations}
where $\boldsymbol\delta=\left[\delta^\mathrm{c}_1,\cdots,\delta^\mathrm{c}_K;\delta^\mathrm{r}_1,\cdots,\delta^\mathrm{r}_K\right]$, $\boldsymbol\zeta=[\zeta^\mathrm{c}_1,\cdots,\zeta^\mathrm{c}_K;\zeta^\mathrm{r}_1,\cdots,\zeta^\mathrm{r}_K]$, and $\boldsymbol\mu=\left[\mu^\mathrm{c}_1,\cdots,\mu^\mathrm{c}_K;\mu^\mathrm{r}_1,\cdots,\mu^\mathrm{r}_K\right]$. Problem \eqref{eq.tbra3} is a convex optimization problem, which can be efficiently solved using the existing optimization toolbox. By iteratively solving problem \eqref{eq.tbra3} until its convergence, we can obtain an asymptotically optimal solution for problem \eqref{eq.tbra}. Algorithm~\ref{algo.tbra} provides the overall procedure for solving the transmit beamforming design and rate allocation subproblem.

\begin{algorithm}
\caption{SCA-Based Transmit Beamforming and Rate Allocation Optimization}
\begin{algorithmic}[1]\label{algo.tbra}
\STATE Initialize $\mathbf{W}^{(0)}$ and $\mathbf{a}^{(0)}$. Set iteration index $i=1$.
\STATE Calculate $\alpha^{(0)},\beta^{(0)},\left(\mu^\mathrm{c}_k\right)^{(0)},\left(\mu^\mathrm{r}_k\right)^{(0)}$ according to \eqref{eq.tbra2.c2}, \eqref{eq.tbra2.c3}, \eqref{cons.RSMA.2.3}, \eqref{cons.rk.5}, respectively.
\REPEAT
\STATE Reformulate problem \eqref{eq.tbra3} with $\mathbf{W}^{(i-1)},\alpha^{(i-1)},\beta^{(i-1)},$ $\left(\mu^\mathrm{c}_k\right)^{(i-1)},\left(\mu^\mathrm{r}_k\right)^{(i-1)}$.
\STATE Solve the reformulated problem \eqref{eq.tbra3} using optimization toolbox, and update $\mathbf{W}^{(i)},\alpha^{(i)},\beta^{(i)},\left(\mu^\mathrm{c}_k\right)^{(i)},\left(\mu^\mathrm{r}_k\right)^{(i)}$ accordingly.
\STATE Set $i=i+1$.
\UNTIL{the objective value of problem \eqref{eq.tbra3} converges.}
\STATE \textbf{Output}: The optimized $\mathbf{W}$ and $\mathbf{a}$.
\end{algorithmic}
\end{algorithm}

\subsection{Semantic Compression Ratio and DC Bias Design}
With given transmit beamforming $\mathbf{W}$ and rate allocation $\mathbf{a}$, the semantic compression ratio and DC bias design subproblem can be given by
\begin{subequations}\label{eq.scradcbd}
    \begin{align}
		\max_{\boldsymbol\rho,B_{\mathrm{DC}}}\quad & \frac{\sum_{k=1}^K \left(a_k+r_k\right)/\rho_k}{P^{\text{comp}} + N U_{\text{LED}} B_{\mathrm{DC}} +\epsilon_3}\tag{\theequation}\\
		\mathrm{s.t.}\hspace{1.3em}
        & P^{\text{comp}} + N U_{\text{LED}} B_{\mathrm{DC}}\leq P^{\max}-\epsilon_3,\label{cons.scradcbd.1}\\
        & \sum_{l=0}^K \left|w_{l,i}\right| \leq B_{\mathrm{DC}}-I_\mathrm{L},\forall i\in\mathcal{N},\label{cons.scradcbd.2}\\
        & \sum_{l=0}^K \left|w_{l,i}\right| \leq I_\mathrm{U}-B_{\mathrm{DC}},\forall i\in\mathcal{N},\label{cons.scradcbd.3}\\
        & \eqref{cons.minrate},\eqref{cons.rho},\notag
    \end{align}
\end{subequations}
where $\epsilon_3=\sum_{l=0}^K\left\|\mathbf{w}_l\right\|^2 + P^{\text{cir}}$ is a constant. The structure of problem \eqref{eq.scradcbd} remains intricate. To gain more insight, we initially focus on the the DC bias, which is a scalar in this problem.

Regarding the DC bias, we have the following theorem.
\begin{theorem}\label{theorem.dc}
    The optimal DC bias in problem \eqref{eq.scradcbd} is
    \begin{equation}
        B_{\mathrm{DC}}^*=\max_{i\in\mathcal{N}}\left\{\sum_{l=0}^K \left|w_{l,i}\right|\right\}+I_\mathrm{L}.
    \end{equation}
\end{theorem}

\begin{proof}
    According to constraint \eqref{cons.scradcbd.2}, we can separate $B_{\mathrm{DC}}$ and rewrite it as
    \begin{equation}\label{eq.theo1}
        B_{\mathrm{DC}} \geq \max_{i\in\mathcal{N}}\left\{\sum_{l=0}^K \left|w_{l,i}\right|\right\}+I_\mathrm{L}.
    \end{equation}
    Similarly, based on constraint \eqref{cons.scradcbd.3}, we have
    \begin{equation}\label{eq.theo2}
        B_{\mathrm{DC}} \leq I_\mathrm{U}-\max_{i\in\mathcal{N}}\left\{\sum_{l=0}^K \left|w_{l,i}\right|\right\}.
    \end{equation}
    Note that the right-hand side of both \eqref{eq.theo1} and \eqref{eq.theo2} are constants in problem \eqref{eq.scradcbd}.
    
    Constraint \eqref{cons.scradcbd.1} can be reformulated as
    \begin{equation}
        B_{\mathrm{DC}} \leq \frac{P^{\max}-\epsilon_3-P^{\text{comp}}}{N U_{\text{LED}}}.
    \end{equation}
    Note that $P^{\text{comp}}$ is related to the semantic compression ratio $\boldsymbol\rho$. To ensure the feasibility of $B_{\mathrm{DC}}$, the following constraint must hold:
    \begin{equation}
        \max_{i\in\mathcal{N}}\left\{\sum_{l=0}^K \left|w_{l,i}\right|\right\}+I_\mathrm{L} \leq \frac{P^{\max}-\epsilon_3-P^{\text{comp}}}{N U_{\text{LED}}}.
    \end{equation}
    Taking into account all the constraints in problem \eqref{eq.scradcbd}, we can obtain the feasible region of $B_{\mathrm{DC}}$ as
    \begin{multline}
        B_{\mathrm{DC}}\in\Bigg[\max_{i\in\mathcal{N}}\left\{\sum_{l=0}^K \left|w_{l,i}\right|\right\}+I_\mathrm{L}, \min\Bigg\{I_\mathrm{U}-\\
        \max_{i\in\mathcal{N}}\left\{\sum_{l=0}^K \left|w_{l,i}\right|\right\},\frac{P^{\max}-\epsilon_3-P^{\text{comp}}}{N U_{\text{LED}}}\Bigg\}\Bigg].
    \end{multline}
    
    Upon examination of the objective function in problem \eqref{eq.scradcbd}, it becomes evident that as the value of $B_{\mathrm{DC}}$ decreases, the value of the objective function increases, under the condition that $\boldsymbol\rho$ remains unchanged. Therefore, the optimal DC bias in problem \eqref{eq.scradcbd} can be written as
    \begin{equation}
        B_{\mathrm{DC}}^*=\max_{i\in\mathcal{N}}\left\{\sum_{l=0}^K \left|w_{l,i}\right|\right\}+I_\mathrm{L},
    \end{equation}
    which is luckily a constant in problem \eqref{eq.scradcbd}.
    $\hfill\blacksquare$
\end{proof}

Substitute the DC bias in problem \eqref{eq.scradcbd} with the obtained optimal solution $B_{\mathrm{DC}}^*$, the problem we need to solve becomes
\begin{subequations}\label{eq.scradcbd2}
    \begin{align}
		\max_{\boldsymbol\rho}\quad & \frac{\sum_{k=1}^K \left(a_k+r_k\right)/\rho_k}{\epsilon_4+\eta\sum_{k=1}^K \max_{m=1,\cdots,M}\left(A_{km}\rho_k+B_{km}\right)}\tag{\theequation}\\
		\mathrm{s.t.}\hspace{1.1em}
        & \eta\sum_{k=1}^K \max_{m=1,\cdots,M}\left(A_{km}\rho_k+B_{km}\right) \leq \epsilon_5,\label{cons.scradcbd2.1}\\
        & \rho_k \leq \frac{a_k+r_k}{R_k},\forall k \in \mathcal{K},\label{cons.scradcbd2.2}\\
        & \rho_k^{\min}\leq\rho_k\leq 1,\forall k \in \mathcal{K},\label{cons.scradcbd2.3}
    \end{align}
\end{subequations}
where $\epsilon_4=N U_{\text{LED}} B_{\mathrm{DC}}^* + \epsilon_3$ and $\epsilon_5=P^{\max} - \epsilon_4$ are constants. Note that the computation overhead function $Q_k\left(\rho_k\right)=\max_{m=1,\cdots,M}\left(A_{km}\rho_k+B_{km}\right)$ represents the pointwise maximum of $M$ linear functions, and this operation preserves convexity. Thus, the numerator and denominator of the objective function in problem \eqref{eq.scradcbd2} are both convex. Furthermore, it is evident that the constraints in problem \eqref{eq.scradcbd2} are all convex. Given the single-ratio fractional structure of the objective function, the Dinkelbach method is adopted to tackle problem \eqref{eq.scradcbd2}.

For clarity, we define
\begin{align}
    f(\boldsymbol{\rho}) &= \sum_{k=1}^K \frac{a_k + r_k}{\rho_k}, \\
    g(\boldsymbol{\rho}) &= \epsilon_4 + \eta \sum_{k=1}^K \max_{m=1,\dots,M} \left(A_{km} \rho_k + B_{km}\right).
\end{align}

In the Dinkelbach method, the transformed problem at the $i$-th iteration is formulated as
\begin{subequations}\label{eq.scradcbd3}
    \begin{align}
		\max_{\boldsymbol\rho}\quad & f(\boldsymbol{\rho}) - \lambda^{(i)} g(\boldsymbol{\rho})\tag{\theequation}\\
		\mathrm{s.t.}\hspace{1.1em}
        & \eta\sum_{k=1}^K Q_k\left(\rho_k\right) \leq \epsilon_5,\label{cons.scradcbd3.1}\\
        & \rho_k \leq \frac{a_k+r_k}{R_k},\forall k \in \mathcal{K},\label{cons.scradcbd3.2}\\
        & \rho_k^{\min}\leq\rho_k\leq 1,\forall k \in \mathcal{K},\label{cons.scradcbd3.3}
    \end{align}
\end{subequations}
where $\lambda$ is a parameter updated at each iteration.

Denote the solution of problem \eqref{eq.scradcbd3} obtained in the $i$-th iteration by $\boldsymbol\rho^{(i)}$. Then, in the subsequent iteration, the parameter $\lambda$ should be updated to
\begin{equation}\label{eq.lambda}
    \lambda^{(i+1)}=\frac{f\left(\boldsymbol{\rho}^{(i)}\right)}{g\left(\boldsymbol{\rho}^{(i)}\right)}.
\end{equation}
This iterative approach continues until $\lambda$ converges.

Nevertheless, problem \eqref{eq.scradcbd3} is still challenging to solve since its objective function is the subtraction of two convex functions, and this operation does not guarantee convexity. Observing the tangible structure of the objective function, it is suitable to use the difference of convex algorithm (DCA) to address problem \eqref{eq.scradcbd3}.

Firstly, we rewrite the objective in problem \eqref{eq.scradcbd3} as
\begin{equation}
    \min_{\boldsymbol\rho} \quad \lambda g(\boldsymbol{\rho})-f(\boldsymbol{\rho}).
\end{equation}
Here, the superscript $(i)$ of $\lambda$ is omitted for the sake of simplicity in notation. Based on the DCA approach, we linearize the convex function $f(\boldsymbol{\rho})$ at each iteration to ensure the convexity of the objective function.

Generally, the first-order Taylor expansion of $f(\boldsymbol{\rho})$ can be expressed as
\begin{equation}
    f(\boldsymbol{\rho}) \approx f\left(\boldsymbol{\rho}^{(j)}\right) + \nabla f\left(\boldsymbol{\rho}^{(j)}\right)^\mathrm{T} \left(\boldsymbol{\rho} - \boldsymbol{\rho}^{(j)}\right),
\end{equation}
where the superscript $(j)$ indicates the iteration of the DCA, and the gradient is given by
\begin{equation}
    \nabla f(\boldsymbol{\rho}) = -\left[\frac{a_k + r_k}{\rho_k^2}\right]_{k=1}^K.
\end{equation}

After linearizing $ f(\boldsymbol{\rho}) $, problem \eqref{eq.scradcbd3} in the $j$-th iteration can be approximated as
\begin{subequations}\label{eq.scradcbd4}
    \begin{align}
		\min_{\boldsymbol\rho}\quad & \lambda g(\boldsymbol{\rho}) - f\left(\boldsymbol{\rho}^{(j)}\right) - \nabla f\left(\boldsymbol{\rho}^{(j)}\right)^\mathrm{T} \left(\boldsymbol{\rho} - \boldsymbol{\rho}^{(j)}\right)\tag{\theequation}\\
		\mathrm{s.t.}\hspace{1em}
        & \eta\sum_{k=1}^K Q_k\left(\rho_k\right) \leq \epsilon_5,\label{cons.scradcbd4.1}\\
        & \rho_k \leq \frac{a_k+r_k}{R_k},\forall k \in \mathcal{K},\label{cons.scradcbd4.2}\\
        & \rho_k^{\min}\leq\rho_k\leq 1,\forall k \in \mathcal{K}.\label{cons.scradcbd4.3}
    \end{align}
\end{subequations}
Problem~\eqref{eq.scradcbd4} is convex and can therefore be efficiently solved using standard convex optimization techniques. Denote the solution obtained at the $(j+1)$-th iteration by $\boldsymbol{\rho}^{(j+1)}$. In each iteration, $\boldsymbol{\rho}^{(j)}$ is updated to $\boldsymbol{\rho}^{(j+1)}$. The iterative process continues until the objective function of problem~\eqref{eq.scradcbd4} converges. The overall procedure for solving the semantic compression ratio and DC bias design subproblem is outlined in Algorithm~\ref{algo.scradcbd}.

\begin{algorithm}
\caption{Semantic Compression Ratio and DC Bias Design Optimization with Dinkelbach Method and DCA}
\begin{algorithmic}[1]\label{algo.scradcbd}
\STATE Calculate the optimal DC bias according to theorem~\ref{theorem.dc}.
\STATE Initialize $\lambda^{(0)}$. Set iteration index $i=1$.
\REPEAT
\STATE Initialize $\boldsymbol{\rho}^{(0)}$. Set iteration index $j=1$.
\REPEAT
\STATE Reformulate problem \eqref{eq.scradcbd4} with $\lambda^{(i-1)}$ and $\boldsymbol{\rho}^{(j-1)}$.
\STATE Solve the reformulated problem \eqref{eq.scradcbd4} using optimization toolbox, and update $\boldsymbol{\rho}^{(j)}$.
\STATE Set $j=j+1$.
\UNTIL{the objective value of problem \eqref{eq.scradcbd4} converges.}
\STATE Obtain solution $\boldsymbol{\rho}^{(i-1)}$.
\STATE Update $\lambda^{(i)}$ according to \eqref{eq.lambda}.
\STATE Set $i=i+1$.
\UNTIL{the objective value of problem \eqref{eq.scradcbd3} converges.}
\STATE \textbf{Output}: The optimized $\boldsymbol{\rho}$ and $B_{\mathrm{DC}}$.
\end{algorithmic}
\end{algorithm}

\subsection{Algorithm Analysis}
Algorithm \ref{algo} outlines the overall optimization approach to solve problem \eqref{eq.pf}. We begin by specifying the initialization and stopping criteria.

\textbf{Initialization:} The alternating optimization algorithm is initialized with a feasible point. Specifically, $\mathbf{W}^{(0)}$ and $\mathbf{a}^{(0)}$ are initialized based on maximum-ratio transmission principles, $B_{\mathrm{DC}}^{(0)}$ is set according to \eqref{cons.LED} with initialized $\mathbf{W}^{(0)}$, and $\boldsymbol\rho^{(0)}$ is set to $\mathbf{1}$, corresponding to no semantic compression. This ensures the algorithm starts from a valid state.

\textbf{Stopping Criteria:} The iterative loops in Algorithms \ref{algo.tbra}, \ref{algo.scradcbd} and \ref{algo} are terminated when the fractional change in their respective objective values between two consecutive iterations, $i-1$ and $i$, is less than a small threshold, $\epsilon$, which is set to $10^{-4}$.

\begin{algorithm}
\caption{Joint Communication and Computation Optimization for VLC-Based PSCom Network}
\begin{algorithmic}[1]\label{algo}
\STATE Initialize $\mathbf{W}^{(0)},B_{\mathrm{DC}}^{(0)},\mathbf{a}^{(0)},\boldsymbol\rho^{(0)}$. Set iteration index $i=1$.
\REPEAT
\STATE With given $B_{\mathrm{DC}}^{(i-1)},\boldsymbol\rho^{(i-1)}$, solve problem \eqref{eq.tbra} with SCA algorithm to obtain solution $\mathbf{W}^{(i)},\mathbf{a}^{(i)}$.
\STATE With given $\mathbf{W}^{(i)},\mathbf{a}^{(i)}$, solve problem \eqref{eq.scradcbd} with Dinkelbach method and DCA to obtain solution $B_{\mathrm{DC}}^{(i)},\boldsymbol\rho^{(i)}$.
\STATE Set $i=i+1$.
\UNTIL{the objective value of problem \eqref{eq.pf} converges.}
\STATE \textbf{Output}: The optimized $\mathbf{W},B_{\mathrm{DC}},\mathbf{a},\boldsymbol\rho$.
\end{algorithmic}
\end{algorithm}

\subsubsection{Convergence Analysis}
Denote the objective value of problem \eqref{eq.pf} at the $i$-th iteration by $V_{\mathrm{obj}}^{(i)}$, the objective value after solving the transmit beamforming design and rate allocation subproblem at the $i$-th iteration by $V_{\mathrm{s1}}^{(i)}$, and the objective value after solving the semantic compression ratio and DC bias design subproblem at the $i$-th iteration by $V_{\mathrm{s2}}^{(i)}$. According to Algorithm~\ref{algo}, we have
\begin{equation}
    V_{\mathrm{obj}}^{(i-1)} \leq V_{\mathrm{s1}}^{(i)} \leq V_{\mathrm{s2}}^{(i)} = V_{\mathrm{obj}}^{(i)},
\end{equation}
which implies that the objective value of problem \eqref{eq.pf} is non-decreasing during the iterations. Furthermore, the objective value is upper-bounded, as it represents the energy efficiency and the total power of the network is constrained. Since the objective value is non-decreasing and upper-bounded, Algorithm~\ref{algo} is guaranteed to converge.

\subsubsection{Computational Complexity Analysis}
The overall computational complexity of the proposed method in Algorithm~\ref{algo} is determined by the iterative solving of its two primary subproblems.
First, for the transmit beamforming and rate allocation subproblem, the complexity is dominated by solving the approximated convex problem \eqref{eq.tbra3}. The cost of this operation is $\mathcal{O}\left(M_1^2 M_2\right)$ \cite{lobo1998applications}, where $M_1 = (N+7)K + N + 4$ is the total number of variables and $M_2 = 9K + 6$ is the total number of constraints. Given that the SCA method requires $N_1$ iterations to converge, the complexity for this subproblem is $\mathcal{O}\left(N_1 N^2 K^3\right)$.
Second, the semantic compression ratio and DC bias subproblem involves solving the convex problem \eqref{eq.scradcbd4} within a Dinkelbach and DCA framework. This has a complexity of $\mathcal{O}\left(M_3^2 M_4\right)$, corresponding to $M_3 = K$ variables and $M_4 = 2K + 1$ constraints. The total complexity for this part is therefore $\mathcal{O}\left(N_2 N_3 K^3\right)$, where $N_2$ and $N_3$ denote the number of iterations for the DCA and the Dinkelbach method, respectively.
Considering that Algorithm~\ref{algo} performs $N_4$ outer-loop iterations, the total computational complexity is the sum of these components, given by $\mathcal{O}\left(N_1 N_4 N^2 K^3 + N_2 N_3 N_4 K^3\right)$.

\section{Simulation Results}\label{Sec:sr}
In the simulation, we consider a setup comprising 4 LEDs and 2 users situated within a room measuring $5\,\mathrm{m} \times 5\,\mathrm{m} \times 3\,\mathrm{m}$. For simplicity, the heights of the LEDs and users are fixed at \SI{3}{m} and \SI{0}{m}, respectively. The spatial coordinates of the LEDs are configured as $(2\,\mathrm{m}, 1.5\,\mathrm{m}, 3\,\mathrm{m})$, $(2\,\mathrm{m}, 3.5\,\mathrm{m}, 3\,\mathrm{m})$, $(3\,\mathrm{m}, 1.5\,\mathrm{m}, 3\,\mathrm{m})$, and $(3\,\mathrm{m}, 3.5\,\mathrm{m}, 3\,\mathrm{m})$, while the user positions are set to $(1.5\,\mathrm{m}, 4\,\mathrm{m}, 0\,\mathrm{m})$ and $(2.5\,\mathrm{m}, 4.5\,\mathrm{m}, 0\,\mathrm{m})$. The parameters for the PSCom model are adopted in accordance with those presented in \cite{ZHAO2024107055}. A comprehensive summary of the system parameters is provided in Table~\ref{tb1}.

\begin{table}[ht]
\renewcommand\arraystretch{1.1} 
\centering
\caption{Main System Parameters}
\begin{tabular}{l|c|c}
    \toprule\hline
    \textbf{Parameter}  & \textbf{Symbol} & \textbf{Value} \\
    \hline
    Number of LEDs & $N$ & $4$ \\
    Number of users & $K$ & $2$ \\
    Physical area of the PD & $A_\mathrm{PD}$ & \SI{1}{cm^2}\\
    Optical filter gain & $G$ & $1$ \\
    Refractive index & $\kappa$ & $1.5$ \\
    Semi-angle of the LED & $\Phi_{1/2}$ & $60^{\circ}$ \\
    FoV of the PD & $\Psi_{\mathrm{PD}}$ & $75^{\circ}$ \\
    Minimum drive current & $I_\mathrm{L}$ & \SI{0.1}{A} \\
    Maximum drive current & $I_\mathrm{U}$ & \SI{1}{A} \\
    Total power budget & $P^{\text{max}}$ & \SI{10}{W} \\
    Circuit power & $P^{\text{cir}}$ & \SI{2}{W} \\
    Forward voltage of the LED & $U_{\text{LED}}$ & \SI{3}{V} \\
    Noise power & $\sigma^2_k$ & \SI{-100}{dBm} \\
    Minimum rate demand & $R_k$ & \SI{1}{bps/Hz} \\
    Computation power coefficient & $\eta$ & $1$ \\
    \hline\bottomrule
\end{tabular}
\label{tb1}
\end{table}
 
\begin{figure}[t]
    \centering
    \includegraphics[width=0.8\linewidth]{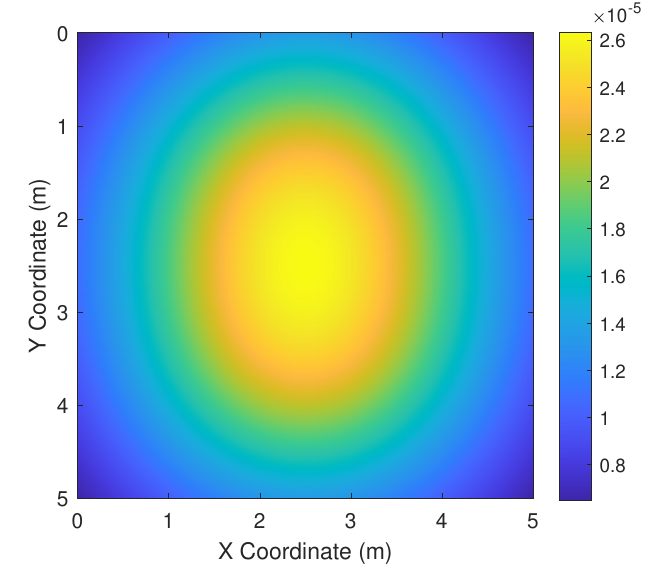}
    \caption{The channel map at zero height under the LED distribution.}
    \label{fig.cm}
\end{figure}

The channel map at floor level corresponding to the deployed LED configuration is depicted in Fig.~\ref{fig.cm}. As shown in the figure, the signal strength exhibits an elliptical attenuation pattern radiating from the central region outward. A pronounced spatial variation in signal intensity is observed across the room, with the highest signal strengths concentrated in the central area and the lowest values occurring at the corners. Notably, the peak signal intensity is approximately three times greater than that at the weakest points. This distribution pattern is primarily attributed to the rectangular arrangement of the LEDs.

To assess the performance of the proposed VLC-based PSCom system, three benchmark schemes are introduced for comparison. The first, denoted as `PSCom-SDMA', replaces RSMA with SDMA. The second, labeled `PSCom-NOMA', utilizes NOMA. The third scheme, `Convention-RSMA', employs RSMA without incorporating semantic compression, transmitting the original data directly. The proposed approach is referred to as `PSCom-RSMA'.

\begin{figure}[t]
    \centering
    \includegraphics[width=\linewidth]{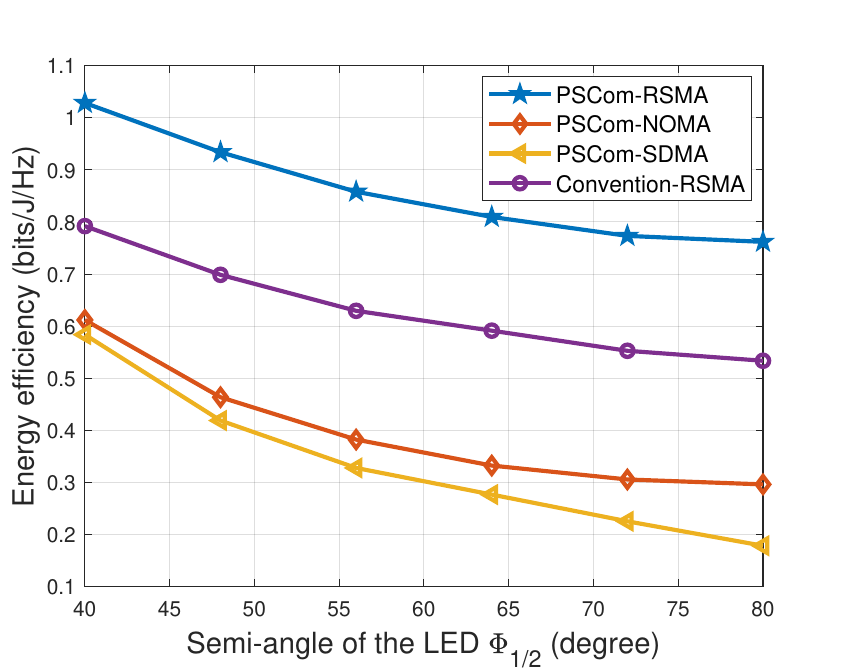}
    \caption{Energy efficiency versus the semi-angle of the LED.}
    \label{fig.Phi_half}
    \vspace{-1.5em}
\end{figure}

Fig.~\ref{fig.Phi_half} illustrates the relationship between energy efficiency and the semi-angle of the LED. It is observed that energy efficiency declines as the LED semi-angle increases across all evaluated schemes. This trend is attributed to the corresponding decrease in the Lambert index, which results in reduced channel gain and, hence, lower energy efficiency. Among the four schemes, the proposed `PSCom-RSMA' consistently achieves the highest energy efficiency, thereby demonstrating the effectiveness of the proposed framework. In contrast, the `PSCom-NOMA' and `PSCom-SDMA' schemes yield inferior performance due to the adoption of NOMA and SDMA, respectively, both of which introduce additional energy consumption, particularly in the transmission of shared knowledge data. The `Convention-RSMA' scheme also underperforms relative to the proposed method, primarily due to its omission of semantic compression, which is an essential mechanism for improving energy efficiency in resource-constrained communication systems.

\begin{figure}[t]
    \centering
    \includegraphics[width=\linewidth]{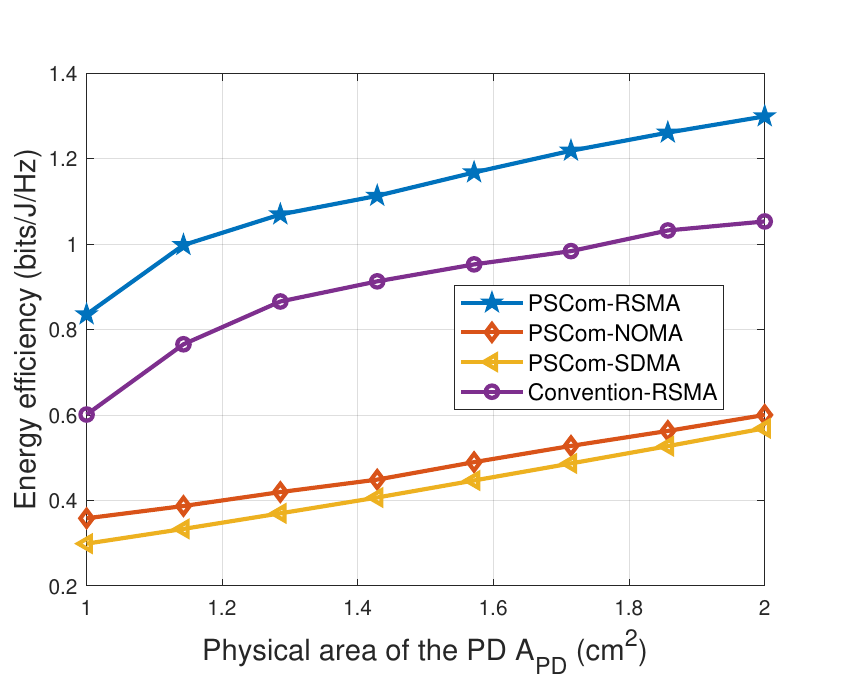}
    \caption{Energy efficiency versus the physical area of the PD.}
    \label{fig.APD}
    \vspace{-1em}
\end{figure}

Fig.~\ref{fig.APD} depicts the relationship between energy efficiency and the physical area of the PD. As shown in the figure, energy efficiency improves across all schemes as the PD area increases. This enhancement is due to the larger PD's ability to capture more incident light, thereby improving the channel gain. Notably, the RSMA-based schemes significantly outperform their NOMA and SDMA counterparts. This performance gain stems from RSMA's capability to encode knowledge update data into the common message, allowing it to be broadcast once to all users. In contrast, both NOMA and SDMA require separate transmissions of this shared data to each user, resulting in higher energy consumption.

\begin{figure}[t]
    \centering
    \includegraphics[width=0.97\linewidth]{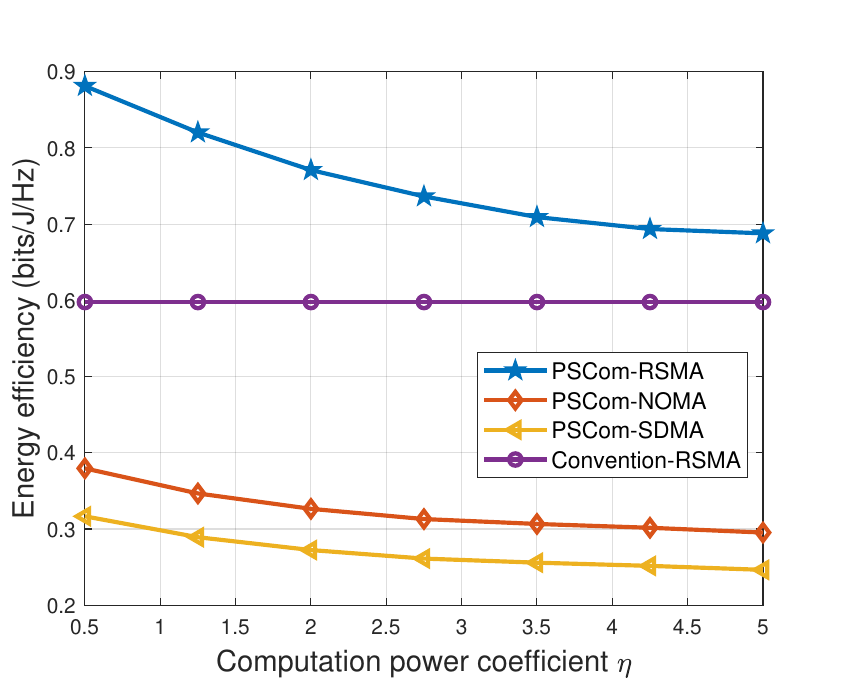}
    \caption{Energy efficiency versus the computation power coefficient.}
    \label{fig.eta}
    \vspace{-1em}
\end{figure}

Fig.~\ref{fig.eta} illustrates the relationship between energy efficiency and the computation power coefficient. As observed, the energy efficiency of the three PSCom-based schemes declines with increasing values of the computation power coefficient. This trend arises because a higher computation power coefficient results in greater energy consumption for semantic compression at a fixed compression ratio, thereby reducing overall energy efficiency. In contrast, the performance of the `Convention-RSMA' scheme remains unaffected by changes in the computation power coefficient, as it does not incorporate semantic compression and, therefore, incurs no associated computational cost.

\begin{figure}[t]
    \centering
    \includegraphics[width=\linewidth]{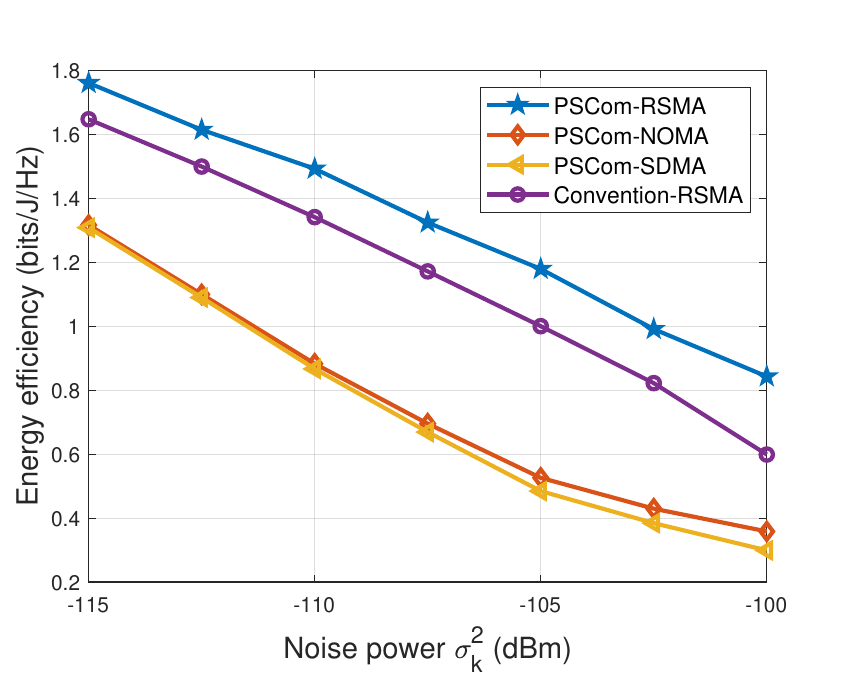}
    \caption{Energy efficiency versus the variance of the AWGN.}
    \label{fig.noise}
    \vspace{-1em}
\end{figure}

Fig.~\ref{fig.noise} presents the relationship between energy efficiency and noise power. As expected, the performance of all schemes deteriorates as noise power increases, owing to the degradation in signal quality and reduced achievable rates. Notably, the performance gap between the `PSCom-NOMA' and `PSCom-SDMA' schemes widens at higher noise levels, with `PSCom-NOMA' exhibiting a relative advantage. This is because NOMA is more resilient to noise due to its power-domain multiplexing and successive interference cancellation, which allows stronger users to decode and subtract weaker users’ signals, thereby partially mitigating the impact of increased noise. In contrast, SDMA relies heavily on spatial separation, which becomes less effective under high noise conditions, leading to more significant performance degradation.

\begin{figure}[t]
    \centering
    \includegraphics[width=\linewidth]{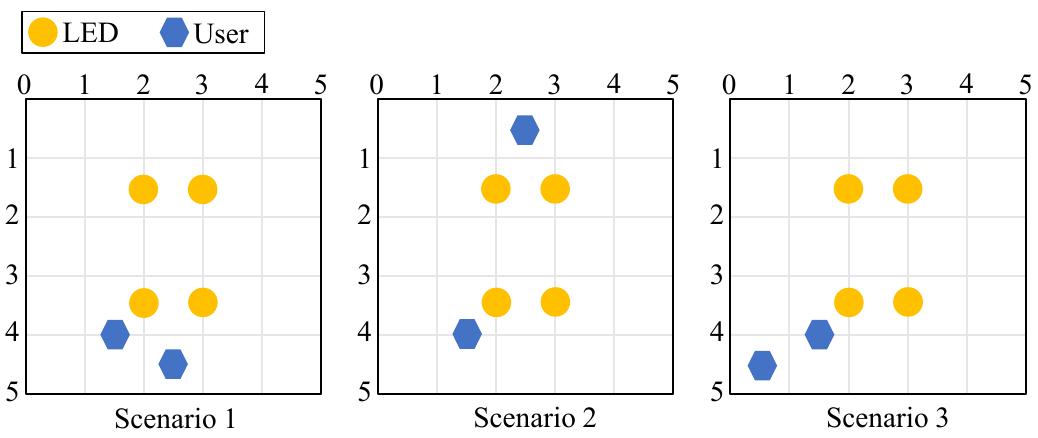}
    \caption{The user distribution in three considered scenarios.}
    \label{fig.scenario}
    \vspace{-1em}
\end{figure}

In addition, three distinct user distribution scenarios are considered to evaluate their impact on system performance, as illustrated in Fig.~\ref{fig.scenario}.

\begin{figure}[t]
    \centering
    \includegraphics[width=\linewidth]{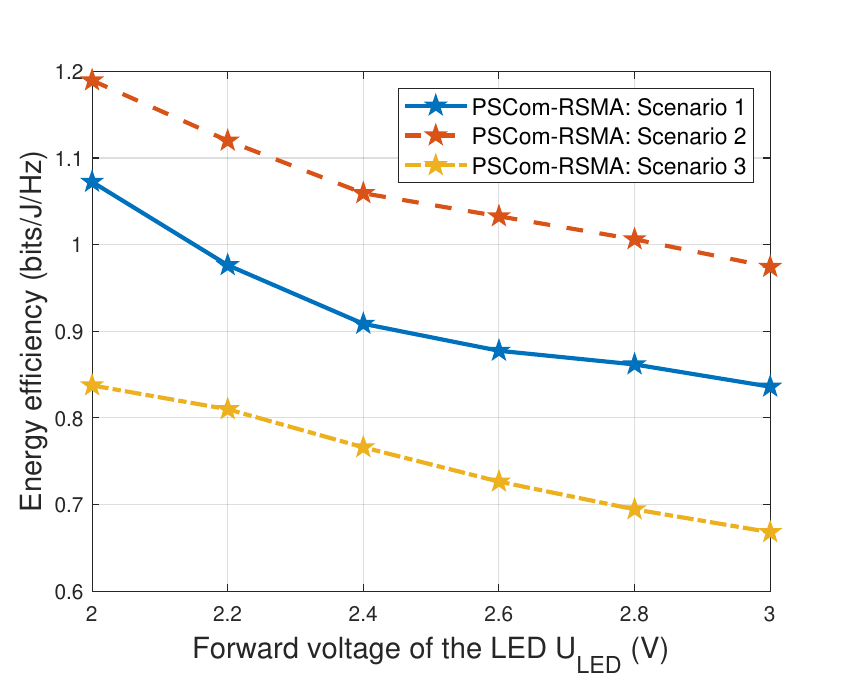}
    \caption{Energy efficiency versus the forward voltage of the LED under different scenarios.}
    \label{fig.s123}
\end{figure}

As shown in Fig.~\ref{fig.s123}, scenario 2 achieves the highest energy efficiency among the three evaluated configurations, followed by scenario 1, while scenario 3 yields the lowest performance. The superior performance of scenario 2 is primarily due to reduced inter-user interference resulting from the increased spatial separation between the users. In contrast, although scenario 3 maintains the same inter-user distance as scenario 1, it suffers from lower channel gains due to suboptimal user placement relative to the LEDs, leading to diminished energy efficiency.

\section{Conclusion}\label{Sec:c}
This paper has developed an energy-efficient framework for VLC-based PSCom networks. Our approach jointly optimizes both communication and computation resources by leveraging RSMA to handle semantic data and shared knowledge. The formulated non-convex maximization problem was effectively tackled by a proposed alternating optimization algorithm that integrates the SCA and Dinkelbach methods. Simulation results have validated the superior performance of our method, highlighting its effectiveness in resource-constrained semantic communication scenarios.

Future work could extend this framework by incorporating more realistic channel models with NLoS and mobility effects. Additionally, as an RSMA-based system, its performance is susceptible to ill-conditioned channels from user clustering, and the PSCom model's integrity depends on successful knowledge synchronization. We also aim to integrate end-to-end semantic quality metrics to directly optimize the trade-off between power consumption and task-oriented reconstruction quality.

\bibliographystyle{IEEEtran}
\bibliography{ref}

\end{document}